%% file: main.tex
\documentclass[10pt,letterpaper,conference]{IEEEtran}

\usepackage{latex_resources/saurav_macros}
\usepackage{pgfplotstable}
\pgfplotsset{compat=1.18} 
\usepgfplotslibrary{colormaps}
\usepgfplotslibrary{fillbetween}
\usepgfplotslibrary{groupplots} 
\pgfplotsset{
    colormap={redblue}{
        rgb255(0cm)=(128,0,0);
        rgb255(1cm)=(255,0,0);
        rgb255(2cm)=(255,255,0);
        rgb255(3cm)=(100,255,0);
        rgb255(4cm)=(0,255,255);
        rgb255(5cm)=(0,0,180);
    }
}
\pgfplotsset{
  colormap={redblue}{
    rgb255(0cm)=(215, 48, 39);
    rgb255(1cm)=(244, 109, 67);
    rgb255(2cm)=(253, 174, 97);
    rgb255(3cm)=(254, 224, 144);
    rgb255(4cm)=(224, 243, 248);
    rgb255(5cm)=(171, 217, 233);
    rgb255(6cm)=(116, 173, 209);
    rgb255(7cm)=(69, 117, 180);
  }
}
\usepackage[caption=false,font=footnotesize]{subfig}
\usepackage{algorithmic}

\hypersetup{
	colorlinks,
	linkcolor={mLinkBlue},
	citecolor={mLinkBlue},
	urlcolor={mLinkBlue}
}

\def\bblambda{\boldsymbol{\lambda}}
\newcommand*{\fn}[2]{\ensuremath{#1\left(#2\right)}}
\newcommand*{\fnXt}{\fn{\mv X}{t}}

\title{\vspace{6mm}\Large \bf Constrained Learning for Decentralized Multi-Objective Coverage Control}

\IEEEoverridecommandlockouts
\author{Juan Cervi\~no$^*$, Saurav Agarwal$^*$, Vijay Kumar, and Alejandro Ribeiro%
	\thanks{$^*$Equal contributions.
    The authors are with the University of Pennsylvania}%
  \thanks{Juan Cervi\~no has moved to the Massachusetts Institute of Technology.}%
  \thanks{E-mail:{\tt\{jcervino,sauravag,kumar,aribeiro\}@upenn.edu}.}%
\thanks{Source code: \url{https://github.com/KumarRobotics/CoverageControl}.}%
\thanks{Supported by the grant ARL DCIST CRA W911NF-17-2-0181.}%
}%

\begin{document}
\bstctlcite{IEEEexample:BSTcontrol}
\maketitle
\begin{abstract}
The multi-objective coverage control problem requires a robot swarm to collaboratively provide sensor coverage to multiple heterogeneous importance density fields~(IDFs) simultaneously.
We pose this as an optimization problem with constraints and study two different formulations:
(1)~Fair coverage, where we minimize the maximum coverage cost for any field, promoting equitable resource distribution among all fields; and
(2)~Constrained coverage, where each field must be covered below a certain cost threshold, ensuring that critical areas receive adequate coverage according to predefined importance levels.
We study the decentralized setting where robots have limited communication and local sensing capabilities, making the system more realistic, scalable, and robust.
Given the complexity, we propose a novel decentralized constrained learning approach that combines primal-dual optimization with a Learnable Perception-Action-Communication (LPAC) neural network architecture.
We show that the Lagrangian of the dual problem can be reformulated as a linear combination of the IDFs, enabling the LPAC policy to serve as a primal solver.
We empirically demonstrate that the proposed method
\emph{(i)}~significantly outperforms state-of-the-art decentralized controllers by 30\% on average in terms of coverage cost,
\emph{(ii)}~transfers well to larger environments with more robots, and
\emph{(iii)}~is scalable in the number of IDFs and robots in the swarm.
\end{abstract}
\section{Introduction}%
\label{sc:introduction}%
\input{intro}

\subsection{Related Work}%
\label{sc:related_work}%
\input{related_work}

\section{Decentralized Multi-Objective Coverage}%
\label{sc:problem_formulation}%
\input{problem_formulation}

\section{Approach: Primal-Dual and LPAC Loops}%
\label{sc:duality}%
\input{dual.tex}

\subsection{Perception-Action-Communication Loops}
\input{lpac.tex}

\section{Experiments}%
\label{sc:experiments}%

\input{experiments}

\section{Conclusions}%
\label{sc:conclusions}%
\input{conclusions}

\newpage
\bibliographystyle{IEEEtran}
\bibliography{bib}  





\end{document}

%% file: intro.tex

In a multi-objective coverage control problem, an environment is characterized by \emph{importance density fields} (IDFs) representing the relative significance of different regions, each capturing a distinct environmental aspect from various information sources.
Depending on the requirements, one may use a \emph{fair coverage} strategy for equitable resource distribution across all fields or a \emph{constrained coverage} strategy to satisfy specific constraints associated with individual fields.

\textbf{Fair coverage} minimizes the maximum coverage cost for any field, promoting equal resource allocation.
For instance, in a flooding scenario, a swarm of robots is deployed to provide sensor coverage over a large urban area.
While the central city region is critical, the outskirts also need monitoring.
A single importance field might cause robots to cluster in the city, neglecting other areas.
A fair coverage strategy ensures robots are distributed across the entire region, providing adequate coverage to all areas, which is crucial in humanitarian applications where resources must be fairly allocated~\cite{malenciafair}.
Figure~\ref{fig:main_env} shows a fair coverage solution.

\begin{figure}[htbp]
  \centering
  \includegraphics[width=1.00\columnwidth]{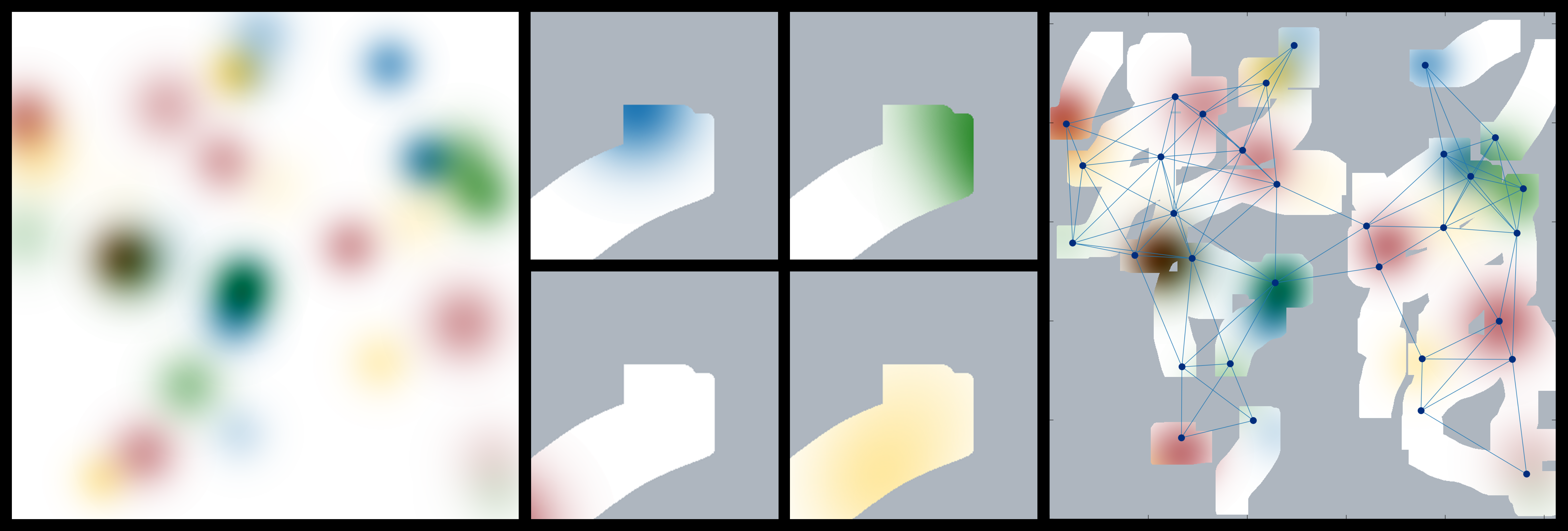}
  \label{fig:env}
  \caption{Multi-objective coverage control on an environment with four importance density fields (IDFs), shown in different colors.
    Robots make localized observations based on their limited sensing range (middle).
 The right figure shows the cumulative observations of all robots, their positions, and the communication graph. The proposed primal-dual algorithm performs dual updates, which are used to re-weight the IDFs, and the LPAC policy computes the velocity actions for the robots.}
  \label{fig:main_env}
\end{figure}

\textbf{Constrained coverage} requires each field to be covered below a certain cost threshold, ensuring critical areas receive coverage according to predefined levels.
In wildfire monitoring, defining a single field for all relevant information is challenging.
Instead, multiple importance fields---such as vegetation density, proximity to habitation, and historical data---provide a comprehensive environmental view.
Setting a cost threshold for each field ensures minimum information acquisition for each aspect appropriately.

To address these optimization problems, we develop a novel decentralized constrained learning approach that combines primal-dual optimization with a \emph{learnable perception-action-communication} (LPAC) loops~\cite{agarwal2024lpac}.
Specifically, we employ a primal-dual algorithm utilizing the Lagrangian dual, which periodically updates the dual variables, and the LPAC policy serves as the primal solver.
The LPAC architecture uses convolutional neural networks (CNNs) to extract features from the robots' local perceptions, graph neural networks (GNNs) to model inter-robot communication, and a shallow multilayer perceptron (MLP) to output robot velocity actions.
The GNN is the primary collaborative component in LPAC as it decides \emph{what} information to communicate and \emph{how} to use received information with its own perception.
This design enables robots to make informed decisions based on local information and limited communication, ensuring scalability and efficiency in decentralized settings.

The primary \textbf{contributions} of this paper are as follows:\\
\emph{(i)}~A novel approach that combines primal-dual optimization with the LPAC architecture for decentralized fair and constrained multi-objective coverage problems with limited communication and sensing robot capabilities.\\
\emph{(ii)}~The Lagrangian of the dual problem is reformulated as a linear combination of IDFs, allowing the multi-objective coverage control problem to be recast as a dynamically re-weighted single-objective problem.\\
\emph{(iii)}~Extensive empirical studies demonstrate the effectiveness of the proposed method compared to existing centroidal Voronoi tessellation methods.
Our approach achieves an average improvement of {$30\%$} in coverage cost, with peak improvements of up to {100\%}.\\
\emph{(iv)}~We establish the scalability and transferability of the proposed method in terms of (a)~the number of robots, (b)~the size of the environment, and (c)~the number of IDFs.

An interesting feature of our approach is that the LPAC policy is trained on instances of a single-objective coverage control problem and then applied to the multi-objective problem without further training.
Given the promising results for the coverage control problem, our approach can potentially provide a robust framework applicable to other multi-robot coordination tasks with constraints.

%% file: related_work.tex
\textbf{Coverage Control} is widely studied in robotics with applications in mobile networking~\cite{hexsel11}, surveillance~\cite{doitsidis2012optimal}, and target tracking~\cite{pimenta2010simultaneous}.
The foundations of decentralized control algorithms for robots with limited sensing and communication capabilities were given by Cort\'{e}s et al.~\cite{Cortes05}.
The algorithms iteratively perform \emph{(i)}~a centroidal Voronoi tessellation (CVT)~\cite{Du99,Edelsbrunner85} of the environment and assigns each robot to their respective Voronoi cell, and \emph{(ii)}~computes the centroid of each cell and moves the robot towards it.
The algorithms are also referred to as variants of the Lloyd's algorithm~\cite{Lloyd82} from quantization theory.
However, there has been a limited study on coverage control with multiple IDFs.

Socially fair coverage control with multiple IDFs arising from population groups was studied by Malencia et al.~\cite{malenciafair}.
By using \textrm{LogSumExp} approximation, a relaxation of the min-max problem was proposed, which led to a control law.

Our approach tackles multi-objective coverage problems in a principled way through duality, recasting them as a single objective coverage control and therefore allowing any coverage control policy to generalize to the multi-objective case without the need for retraining or redesigning controllers.

\textbf{Graph Neural Networks} are cascaded architectures composed of layers of graph convolutions and pointwise non-linearities \cite{gama2019convolutional}.
GNNs have shown impressive results in problems like weather prediction \cite{lam2022graphcast}, recommendation systems \cite{he2020lightgcn}, and physics \cite{sanchez2020learning}.
Given their graph convolutional structure, GNNs have been used for decentralized problems in various robotics applications such as multi-robot active information acquisition \cite{tzes2023graph}, coverage control \cite{gosrich2022coverage,agarwal2024lpac}, and path planning \cite{tolstaya2020learning,li2020graph,li2021message}.
GNNs are well-suited for robotic applications due to their stability to perturbations on the graph \cite{wang2024generalization,gama2020stability}, transferability to larger-scale graphs \cite{ruiz2020graphon}, and generalization to unseen graphs \cite{wang2024manifold}.

\textbf{Constrained learning} is a machine learning technique that tackles problems in which more than one objective needs to be satisfied jointly \cite{chamon2022constrained,chamon2020probably}.
Notable applications include federated learning \cite{shen2021agnostic}, passing through robustness \cite{robey2021adversarial} and smoothness \cite{cervino2023learning}, and active learning \cite{elenter2022lagrangian}.
Constrained learning has found applications in GNNs for improving their stability \cite{cervino2022training} and robustness \cite{arghal2022robust}.
In robotics, constrained learning has been applied to safe reinforcement learning \cite{paternain2022safe,castellano2022reinforcement,rozada2024deterministic} and legged robot locomotion \cite{kim2024not}.

%% file: problem_formulation.tex
The coverage problem is defined on a $d$-dimensional environment $\mc W\subset\mathbb R^2$ with a homogeneous set of $N$ robots $\mc V=\{1,\dots,N\}$.
The position of the robot at time $t$ is denoted by $\mv p_i(t)\in \mc W$.
Robots take control actions $\mv u_i(t)\in \mathbb R^2$ to move in the environment and follow a single integrator dynamics every time step $\Delta t$:
\begin{align}
  \mv p_i(t+\Delta t)=\mv p_i(t)+\Delta t \mv u_i(t).
\end{align}

In a decentralized setting, each agent makes observations of its local environment and communicates with nearby robots to exchange information.
We define a \textit{communication graph} graph $\mc G=(\mc V,\mc E)$, where $\mc V$ is the set of robots and $\mc E$ is the set of edges that represents a communication link between the robots.
The communication link $(i,j)\in\mc E$ exists if the distance between the robots $i$ and $j$ is less than the communication radius $r_c$, i.e. $\| \mv p_i- \mv p_j\|\leq r_c$.
We assume the communication occurs both ways and, therefore, the graph is undirected.
A robot may communicate an abstract representation of its state with its neighbors, and the communication graph is used to determine the neighborhood of each robot.
Decentralized policies are executed by each robot independently, and do not require a centralized server thus alleviating the communication and scalability issues. 

In coverage control, the objective is to provide sensor coverage based on the relative importance of each point in the environment.
We define an \emph{importance density field} (IDF) $\Phi:\mc W\to\mathbb R^+$, which represents the importance of each point in the environment.
The cost function $\mc J$ is thus defined as:
\begin{equation}
  \mc J(\mv X(t))=\int_{\mv q\in \mc W}\min_{i\in \mc V} f(\|\mv p_i-\mv q\|)\Phi(\mv q)d(\mv q),\label{eqn:coverage_cost}
\end{equation}
where $f$ is a non-decreasing function, and a common choice is $f(x)=x^2$.
The state of the system is given by $\mv X(t)=\{\mv p_i(t)\}_{i\in \mc V}$.
Assuming that two robots cannot occupy the same point, we can assign each robot a distinct portion of the environment to cover based on the Voronoi partitions \cite{de2000computational}, and the cost function can be rewritten as:
\begin{equation}
  \mc J\left(\mv X(t)\right) = \sum_{i=1}^N \int_{\mv q \in P_i} f(\|\mv p_i(t) - \mv q\|) \Phi(\mv q)\,d\mv q,
  \label{eq:coverage_voronoi}
\end{equation}
where $P_i$ is the Voronoi partition of robot $i$.

\subsection{Multi-Objective Coverage}
We formulate two multi-objective coverage problems with $M$ independent IDFs $\Phi_m:\mc W\to\mathbb R^+$ for $m=1,\dots,M$:

\noindent\emph{(i)}~Fair coverage minimizes the maximum coverage cost of any IDF.
\emph{(ii)}~Constrained coverage minimizes the cost function of one IDF while ensuring that the cost functions of the other IDFs are below a certain threshold.

\textbf{Fair Coverage:}
Given a team of $N$ robots operating in an environment with $M$ IDFs, the \emph{fair multi-objective coverage problem} (FMCP) is defined as:
\begin{align}
  &P_f=  \min_{\mv X(t),\rho} \quad \rho\label{eqn:epi_multi_objective_coverage}\tag{FMCP}\\
  \text{subject to} \quad &\mc J_m (\mv X(t))\leq \rho \quad \forall\, m=1,\dots,M \nonumber
\end{align}
In the FMCP, all coverage functions need to be below $\rho$.
This is typically called an epigraph problem \cite{boyd2004convex}.

\textbf{Constrained Coverage:}
Given a team of $N$ robots operating in an environment with $M$ IDFs, the \emph{constrained multi-robot coverage problem} (CMCP) is defined as:
\begin{align}
  &P_c= \min_{\mv X(t)} \quad   \mc J_0 (\mv X(t))\label{eqn:multi_objective_coverage}\tag{CMCP}\\
  \text{subject to} \quad & \mc J_m (\mv X(t))\leq \alpha_m \quad \forall\, m=1,\dots,M. \nonumber
\end{align}
There are $M+1$ IDFs that need to be covered.
The $0$-th IDF needs to be minimized, while the other $M$ of them need to be covered up to their corresponding ~$\alpha_m$. 
Note that the $0$-th objective is not necessary, and the formulation can be used to obtain feasible solutions that satisfy the constraints.

%% file: dual.tex
To solve the multi-objective coverage control problem, we resort to the dual domain.
Denote the dual variables by $\lambda_i\geq 0, i \in \{0,\ldots,M\},\bblambda=[\lambda_1,\dots,\lambda_M]$, the state $\fnXt$ by $\mv X_t$, and define the constrained $\mc L_c$ and fair $\mc L_f$ Lagrangians as:
\begin{align}
  \fn{\mc L_f}{\mv X_t, \bblambda} &\coloneq  \sum_{m=1}^M \lambda_m \fn{\mc J_m}{\mv X_t},\,\text{s.t.}\sum_{i=1}^M \lambda_i = 1,\label{eqn:lagrangian_fair}\\
  \fn{\mc L_c}{\mv X_t, \bblambda} &\coloneq \fn{\mc J_0}{\mv X_t} + \sum_{m=1}^M \fn{\lambda_m}{\fn{\mc J_m}{\mv X_t} - \alpha_m}.\label{eqn:lagrangian_constrained}
\end{align}
Note that the variable $\rho$ disappears from the Lagrangian $\mc L_f$, as it is a property of the solution (see \cite[Appendix A]{cervino2023learning}). 

Denoting a general Lagrangian by $\mc L_*, *\in\{f,c\}$, the problem $D_*$ and dual function $d_*(\cdot)$ are 
\begin{equation}
  D_*=\max_{\bblambda}\big( d_*(\bblambda) \coloneq \min_{\mv X_t} \mc L_*(\mv X_t, \bblambda)\big).\label{eqn:dual_problem}
\end{equation}
While the primal constrained optimization problems $P_f$ \eqref{eqn:epi_multi_objective_coverage} and $P_c$ \eqref{eqn:multi_objective_coverage} are computationally challenging, the dual counterparts $D_*$ are linear combinations of convex optimization problems that can be solved in practice.
We now show that the dual problem~\eqref{eqn:dual_problem} is equivalent to a coverage control problem on the linear combination of the IDFs.
%
\begin{proposition}\label{prop:IDF_lambda}
  Given a set of scalars $\lambda_m\geq 0,m=1,\dots, M$, the linear combination of coverage control problems $\sum_{m=1}^M\lambda_m\fn{\mc J_{\phi_m}}{\fnXt}$, is equivalent to a coverage control problem on the linear combinations of IDFs, i.e., $\fn{\mc J_{\phi_\lambda}}{\fnXt}$, with $\phi_\lambda(x)=\sum_{m=1}^M\lambda_m\phi_m(x)$.
\end{proposition}
\begin{proof}
  Let the distance function $\dist(\mv q, t)$ be defined as:
  \begin{equation*}
    \dist(\mv q, t) \coloneq \min_{i\in \mc V}\fn{f}{\|\mv p_i(t)-\mv q\|}.
  \end{equation*}
  By the definition of $\mc J$ in \eqref{eqn:coverage_cost} and the linearity of integration:
  \begin{equation*}
    \begin{split}
      \sum_{m=1}^M \lambda_m \fn{\mc J_m}{\fnXt}&=\sum_{m=1}^M \lambda_m\smashoperator{\int\limits_{\mv q\in\mc W}}  \dist(\mv q, t)\phi_m(\mv q)d\mv q\\
                                                &=\smashoperator{\int\limits_{\mv q\in\mc W}} \dist(\mv q, t) \sum_{m=1}^M \lambda_m \phi_m(\mv q)d\mv q.
    \end{split}
  \end{equation*}
Note that $\sum_{m=1}^M\lambda_m\phi_m(x)$ is a valid IDF as $\lambda\geq 0$. 
By defining $\phi_\lambda(x)=\sum_{m=1}^M\lambda_m\phi_m(x)$, we obtain the result. 
\end{proof}
Proposition~\ref{prop:IDF_lambda} implies that a linear combination of coverage control problems is equivalent to a coverage control problem on the linear combinations of IDFs.
%
\begin{corollary}\label{cor:IDF_lambda}
  The dual function $d_*(\bblambda)$ (cf. \eqref{eqn:dual_problem}) associated with the Lagrangians $\fn{\mc L_f}{\mv X(t),\bblambda}$ of the fair coverage problem $P_f$ (FMCP) and $\mc L_c(\mv X(t), \bblambda)$ of the constrained coverage problem $P_c$ (CMCP) are equivalent to a single objective unconstrained coverage control problem 
with $\phi_{\bblambda}(\mv q)=\sum_{m=1}^M \lambda_m\phi_m(\mv q)$.
\end{corollary}
\begin{proof}
We can apply Proposition \ref{prop:IDF_lambda} to the Lagrangians $\mc L_f$ and $\mc L_c$.
By defining $\phi_{\bblambda}(\mv q)=\sum_{m=1}^M \lambda_m\phi_m(\mv q)$, and noting that this is a valid IDF given that $\lambda\geq 0$, we complete the proof.
For the constrained Lagrangian~$\mc L_c$,  the term $\sum_{m=1}^M\lambda_m\alpha_m$ is a constant and does not take part in the minimization.
\end{proof}
Corollary \ref{cor:IDF_lambda} is important because it allows us to simplify the multi-objective coverage control problem by recasting it as a single objective, thus utilizing a policy $\pi$ for a single coverage control. The case with constraints $\mc L_c$ is equivalent to the minimization problem. 
Under mild assumptions, and both for convex as well as non-convex losses, it can be shown that the constrained optimization problem $P_*$, and the dual counterpart $D_*$ are close \cite{chamon2020probably,chamon2022constrained,elenternear}.

\subsection{Primal-Dual Algorithm For Multi-Objective Coverage}
To obtain the optimal dual variables and solve the dual problem \ref{eqn:dual_problem}, we will resort to a primal-dual algorithm. There are two time frames, the primal and the dual time frame, the dual being slower than the primal, i.e., the value of $\lambda$ remains constant over a $T$ period, upon which it is modified. 

In the primal update, at time step $k$, given dual variables $\bblambda_k=[\lambda^0_k,\dots,\lambda^M_k]$, the objective is to find the action such that the objective $\mc J_{\bblambda_k}(\mv X(t))$ is minimized.
After $T$ steps, the dual variables get updates according to the following rules:
\begin{align}
\tilde \bblambda_{k+1} = [\lambda_{k}^i+\eta s^i(k)]_+ \text{ (dual update),}\label{eqn:dual_update}\\
\bblambda_{k+1} = \mc P(\tilde \bblambda_{k+1} )\text{ (dual projection),}\label{eqn:dual_projection}
\end{align}
where $s^i$ represents the slackness (or constraint satisfaction) of the $i$-th constraint, $[\cdot]_+=\max(0,\cdot)$ element-wise, and $\mc P$ is a projection.

For \eqref{eqn:multi_objective_coverage}, the slackness takes the form:
\begin{align}
s^i_c(k) = \frac{1}{T}\sum_{t=kT}^{(k+1)T}(\mc J_i\left(\mv X(t)\right)-\alpha_i,
\end{align} 
and there is no projection, i.e. $\tilde \bblambda = \bblambda$.

For \eqref{eqn:epi_multi_objective_coverage}, the dual slackness and the projection are:
\begin{align}
s^i_f(k) &= \frac{1}{T}\sum_{t=kT}^{(k+1)T}\mc J_i\left(\mv X(t)\right)\\
\bblambda &= \argmin_{\bblambda} ||\bblambda - \tilde\bblambda||_2\text{ s.t. } \sum_{m=1}^M \lambda_m = 1.
\end{align} 

This projection~\cite{duchi2008projection} is required to satisfy the constraint added to the Lagrangian in \eqref{eqn:lagrangian_fair}.
A succinct explanation of the primal-dual steps can be found in Algorithm~\ref{alg:primal_dual}.

\begin{algorithm}[t]
\caption{Primal Dual Multi-Objective Coverage}\label{alg:primal_dual}
\begin{algorithmic}
  \small
  \STATE \textbf{Input:} Policy $\pi$, initial dual variables $\bblambda$, dual period $T$, dual step size $\eta$, total number of dual steps $K$
  \FOR{$k=1$ to $K$}
  \FOR{$\bar{t} = 1$ to $T$}
  \STATE $t\leftarrow (k-1)T+\bar{t}$
  \STATE Compute action $\mv U(t)\leftarrow\pi(\mv X(t),\phi_{\bblambda})$
  \STATE Perform control actions $\mv U(t)$ using policy $\pi$
  \ENDFOR
  \STATE Update dual variables according to \eqref{eqn:dual_update}
  \STATE Project dual variables according to \eqref{eqn:dual_projection}
  \ENDFOR
\end{algorithmic}
\end{algorithm}

%% file: lpac.tex
Intelligent collaborative and decentralized robotic systems require each robot to independently perceive the environment, communicate information with nearby neighbors, and take appropriate actions by using all the received and sensed information.
This leads to a Perception-Action-Communication (PAC) loop.
In a Learnable PAC (LPAC) architecture~\cite{agarwal2024lpac}, we have a neural network for each module: a Convolutional Neural Network (CNN) for perception, a Graph Neural Network (GNN) for communication, and a Multi-Layer Perceptron (MLP) for action.
We train an LPAC policy using imitation learning, where we use a clairvoyant CVT-based controller to generate target velocity actions for the robots.

\textbf{Environment setup}: Robots are initialized uniformly at random in a square environment with side length \SIm{1024} and a resolution of \SIm{1} per pixel.
For each importance density field (IDF), we generate a fixed number of 2D Gaussian distributions, where the mean is uniformly sampled in the environment, and the variance is sampled from $[40, 60]$ with a limit of $3$ standard deviations to avoid detection from afar.
Each such distribution is scaled by a random factor sampled from $[0.05, 1.00]$ to have a wide range of importance values.

As robots move, they maintain an \textit{importance} map for each IDF to store observations and a \textit{boundary} map to store perceived boundaries.
The observations are limited to a $64 \times 64$ window around the robot.
The robots follow a single integrator dynamics with a maximum speed of \SIvel{1}.

We now describe the three modules of the architecture.

The \textbf{perception} module comprises a CNN that takes as input a four-channel image representing the IDF, the boundary, and the relative coordinates of nearby robots.
A local region of $256 \times 256$ pixels is extracted from the importance maps and the boundary map, and they are scaled down to $32 \times 32$ pixels; this forms the first two channels of the input.
The relative coordinates of the nearby robots are scaled by the communication range and are added as the last two channels.

The CNN consists of three sequences, each comprising a convolutional layer followed by batch normalization~\cite{IoffeS15BatchNorm} and a leaky ReLU activation function.
Each convolutional layer uses a $3\times3$ kernel with a stride of one and zero padding, producing 32 output channels.
The output of the final sequence is flattened and passed through a linear layer with a leaky ReLU activation to generate a 32-dimensional feature vector, which is sent to the communication module.

The \textbf{communication} module receives the feature vector from the perception module and propagates it through the communication graph using a Graph Neural Network (GNN).
The GNN computes a message for each robot that needs to be communicated to other agents.
Furthermore, the GNN aggregates the received messages and combines them with the local feature vector to generate the final output, which is passed to the action module.

Our GNN architecture is a composition of $5$ \textit{graph convolution filters}~\cite{RuizGR21} with ReLU as the pointwise nonlinearity.
Each convolution is parameterized by $k=3$ hops, and each layer has $512$ hidden units.
The shift operator $S$ is the normalized adjacency matrix of the communication graph,
	$\mv S = \mv D^{-1/2} \mv A \mv D^{-1/2}$,
$\mv A$ is the adjacency matrix, and $\mv D$ is the diagonal degree matrix.
Each layer is given by:
\begin{equation}
	\mv Z_l = \sum_{k=0}^K (\mv S)^k \mv X_{l-1} \mv H_{lk}, \quad \mv X_l = \sigma(\mv Z_l).
\end{equation}
Here, $\mv X_l$ is the output of the $l$-th layer, $\mv H_{lk}$ is the weight matrix of the $k$-th filter in the $l$-th layer, and $\sigma$ is the ReLU activation function.

The \textbf{action} module is the last link in the chain and is in charge of generating the local velocity control action for the robot.
The MLP has two layers with 32 units each, and the final output of the MLP is processed by a linear layer to generate the $x$ and $y$ components of the velocity action.

\subsection{Imitation Learning}
We use imitation learning to train the LPAC architecture to mimic the behavior of the clairvoyant algorithm.
The training is performed using the Adam optimizer~\cite{KingBa15} in Python using PyTorch~\cite{PyTorch} and PyTorch Geometric~\cite{PyTorchGeometric}.
We use a batch size of 100 and train the network for 100 epochs, with a learning rate of $10^{-4}$ and a weight decay of $10^{-3}$.
We use the mean squared error (MSE) loss as the loss function, where the target is the output of the clairvoyant algorithm, and the prediction is the output of the LPAC architecture.
The model with the lowest validation loss is selected as the final model.

\begin{figure}[htbp]
  \centering
  \input{plots/baseline_coverage.tex}
  \caption{Fair coverage control problem with $32$ robots and $4$ IDFs in a $1024\times 1024$ \SImsqr{} environment. The coverage cost increase in comparison to the clairvoyant is $81\%$, $246\%$, $250\%$, and $258\%$ for LPAC, Centralized, Decentralized, and SFCC \cite{malenciafair} respectively. 
    The LPAC policy outperforms the centralized CVT and SFCC approaches.\label{fig:fair_coverage}}
\end{figure}

%% file: plots/baseline_coverage.tex
\begin{tikzpicture}[trim axis right]
	\footnotesize
	\tikzset{
		mylabel/.style={pos=0.85, above, yshift=-2.5pt}
	}
	\pgfplotstableread[col sep=comma]{./plots/data/baseline_fair_coverage.csv}{\costdata}
	\begin{axis}[
		width=1.05\columnwidth,
		height=0.682\columnwidth,
		xlabel={Time},
    x unit={s},
		ylabel={Normalized Coverage Cost},
		xlabel style={font=\footnotesize},
    ylabel style={font=\footnotesize, inner sep=-0.4pt},
		legend style={font=\scriptsize, at={(0.50,1.15)},anchor=north},
		legend columns=-1,
    grid=both,
		xmin=-5,
		xmax=255,
		ymin=0.5,
		ymax=6.5,
		]
		\def\controller{clair}
		\addplot[mark=none, smooth,draw=mSteelGray, thick]  table[x=Step, y=\controller, col sep=comma] {\costdata} node[mylabel, yshift=+3.0pt] {\textcolor{mSteelGray!80}{\footnotesize Clairvoyant}};
		\addplot[name path=upper,draw=mSteelGray!70, dashed, thin] table[x=Step, y expr=\thisrow{\controller}+\thisrow{\controller_std}, col sep=comma] {\costdata};
		\addplot[name path=lower,draw=mSteelGray!70, dashed, thin] table[x=Step,y expr=\thisrow{\controller}-\thisrow{\controller_std}, col sep=comma] {\costdata};
		\addplot[fill=mSteelGray!10, fill opacity=0.5] fill between[of=upper and lower];

		\def\controller{centralized_malencia}
    \addplot[mark=none, smooth, draw=mDarkRed, thick]  table[x=Step, y=\controller, col sep=comma] {\costdata} node[mylabel, yshift=+1.0pt] {\textcolor{mDarkRed!80}{\footnotesize SFCC~\cite{malenciafair}}};
		\addplot[name path=upper,draw=mDarkRed!80, thin, dashed] table[x=Step, y expr=\thisrow{\controller}+\thisrow{\controller_std}, col sep=comma] {\costdata};
		\addplot[name path=lower,draw=mDarkRed!80, thin, dashed] table[x=Step,y expr=\thisrow{\controller}-\thisrow{\controller_std}, col sep=comma] {\costdata};
		\addplot[fill=mDarkRed!40, fill opacity=0.5] fill between[of=upper and lower];

		\def\controller{centralized}
		\addplot[mark=none, smooth,draw=mGreen, thick]  table[x=Step, y=\controller, col sep=comma] {\costdata} node[mylabel, yshift=-7.5pt] {\textcolor{mGreen}{\footnotesize Centralized CVT}};
		\addplot[name path=upper,draw=mGreen!70, thin, dashed] table[x=Step, y expr=\thisrow{\controller}+\thisrow{\controller_std}, col sep=comma] {\costdata};
		\addplot[name path=lower,draw=mGreen!70, thin, dashed] table[x=Step,y expr=\thisrow{\controller}-\thisrow{\controller_std}, col sep=comma] {\costdata};
		\addplot[fill=mGreen!40, fill opacity=0.8] fill between[of=upper and lower];

		\def\controller{LPAC}
		\addplot[mark=none, smooth,draw=mDarkBlue, thick]  table[x=Step, y=\controller, col sep=comma] {\costdata} node[mylabel, yshift=+3.0pt] {\textcolor{mDarkBlue!80}{\footnotesize LPAC}};
		\addplot[name path=upper,draw=mDarkBlue!40, dashed, thin] table[x=Step, y expr=\thisrow{\controller}+\thisrow{\controller_std}, col sep=comma] {\costdata};
		\addplot[name path=lower,draw=mDarkBlue!40, dashed, thin] table[x=Step,y expr=\thisrow{\controller}-\thisrow{\controller_std}, col sep=comma] {\costdata};
		\addplot[fill=mDarkBlue!40, fill opacity=0.5] fill between[of=upper and lower];

	\end{axis}
\end{tikzpicture}

%% file: experiments.tex
This section provides an empirical study of the two problems \textit{fair coverage}~\eqref{eqn:epi_multi_objective_coverage} and \textit{constrained coverage}~\eqref{eqn:multi_objective_coverage} in the context of multi-objective coverage control.
We show that the proposed primal-dual algorithm with LPAC policy is:
\emph{(i)}~an efficient solver for the multi-objective coverage control problem,
\emph{(ii)}~generalizes to a wide range of robot parameters, and
\emph{(iii)}~scales with number of robots, number of constraints, and environment size.

We compare our method against four baseline controllers: socially-fair coverage control (SFCC)~\cite{malenciafair}, decentralized and centralized centroidal Voronoi tessellation (CVT) algorithms, and the clairvoyant CVT algorithm used for imitation learning.
SFCC \cite{malenciafair} is a centralized gradient-based algorithm that operates on the $\textbf{LogSumExp}$ of the IDFs.
The decentralized and centralized CVT algorithms operate on the partial information about the IDF gathered by the robots along their trajectories (see~\cite[Section V.A]{agarwal2024lpac}).

 Note that in our proposed approach, the LPAC policy has the same limited information as the decentralized CVT algorithm:
 \emph{(i)}~robots have a limited sensing radius to gather information about the IDF,
 \emph{(ii)}~robots can communicate with each other within a limited radius, and
 \emph{(iii)}~robots compute actions independently based on the information gathered from the environment and the communication with nearby robots.
 In contrast, the centralized CVT algorithm has information about the positions of all robots and their observations of the IDF and can compute actions for all robots simultaneously.
 The clairvoyant algorithm has perfect knowledge of the IDF and the positions of the robots resulting in near-optimal actions.

We begin by considering the fair coverage control problem \eqref{eqn:epi_multi_objective_coverage}.
\fgref{fig:fair_coverage} presents the maximum coverage cost across all IDFs averaged over $100$ randomly generated environments using the LPAC, centralized CVT, and clairvoyant controllers.
The LPAC outperforms both the decentralized and centralized CVT controller by $50\%$.
Moreover, the standard deviation bands of the LPAC controller are lower than the mean of the centralized CVT controller, indicating that the LPAC controller consistently performs better across different environments.
The result illustrates the benefits of using GNNs for multi-objective coverage control, even when accessing partially decentralized information. 

\begin{figure}[htbp]
  \centering
  \input{plots/comms_sensors.tex}
  \caption{Generalization study for the fair coverage problem for varying communication radius (192, 256, 320)
    and sensor size (64, 96, 128), averaged over $100$ environments.
  The LPAC consistently outperforms the centralized CVT for all configurations.\label{fig:comm_radius}}
\end{figure}
\fgref{fig:comm_radius} shows the performance of the proposed LPAC approach with respect to the centralized CVT controller as the communication radius and the sensor size of the robots are varied.
The centralized algorithm is agnostic to the communication radius, and its performance increases significantly with sensor size.
While the LPAC approach is a little sensitive to these parameters, it consistently outperforms the centralized CVT for all configurations.

\begin{figure}[htbp]
  \centering
  \input{plots/case_study.tex}
  \caption{Case study: Correlation of maximum coverage cost (top row) and dual variables (bottom row) per IDF for the environment shown in \fgref{fig:main_env}. Each color represents a different IDF. The primal-dual algorithm performs dual updates such that the least covered IDF has more weight. The LPAC policy drives the robots efficiently to reduce the coverage cost.\label{fig:fair_objective_dual_variables}}
\end{figure}

\fgref{fig:fair_objective_dual_variables} shows the correlation of the maximum coverage cost and the dual variables for the environment in~\fgref{fig:main_env}.
The upper row shows how the performance across IDFs improves as time progresses. 
The clairvoyant and LPAC switch the largest dual variable from blue to red, at around $180 s$, thereby lowering the coverage cost for both IDFs, making the solution more equitable.
However, the centralized CVT struggles to lower the cost for the blue IDF.
The dual value is proportional to the difficulty of satisfying the constraint, which translates into minimizing the cost of the corresponding IDF.

\begin{table}[tbp]
  \renewcommand{\arraystretch}{1.2}
  \centering
  \caption{Fair Coverage: Varying Number of Robots and IDFs}%
  \label{tab:centralized_vs_lpac_fair_robot_ablation}%
  \begin{tabular}{cccccc}
    \toprule
    &\multicolumn{1}{c}{\multirow{2}{*}{}} &\multicolumn{2}{c}{\# Best Environments} & \multicolumn{2}{c}{Average Max Objective }\\\cmidrule(l){3-4} \cmidrule(l){5-6}
    &\multicolumn{1}{c}{} & \multicolumn{1}{r}{Centralized CVT} & \multicolumn{1}{r}{LPAC} & \multicolumn{1}{r}{Centralized  CVT} & \multicolumn{1}{r}{LPAC}\\\midrule
    \parbox[c]{2mm}{\multirow{8}{*}{\rotatebox[origin=c]{90}{\# Robots}}}& 8 & 32 & 68 & 2.839&2.763 \\
                                                                    & 16 & 15 & 85 & 2.680&2.191 \\
                                                                    & 24 & 4 & 96 & 2.559&1.687 \\
                                                                    & 32 & 6 & 94 & 2.251&1.386 \\
                                                                    & 40 & 2 & 98 & 2.167&1.191 \\
                                                                    & 48 & 6 & 94 & 1.842&1.163 \\
                                                                    & 56 & 3 & 97 & 1.744&0.981 \\
                                                                    & 64 & 3 & 97 & 1.608&0.983 \\
    \midrule
    \parbox[c]{2mm}{\multirow{7}{*}{\rotatebox[origin=c]{90}{\# IDFs}}}        &4 & 6 & 94 & 2.251&1.386 \\
            &5 & 6 & 94 & 3.137&1.979 \\
            &6 & 5 & 95 & 4.144&2.486 \\
            &7 & 4 & 96 & 5.072&3.214 \\
            &8 & 4 & 96 & 6.144&3.792 \\
            &9 & 1 & 99 & 7.060&4.667 \\
            &10 & 0 & 100 & 8.104&5.037 \\
    \bottomrule
  \end{tabular}
\end{table}

\begin{figure*}[htbp]
  \centering
  \subfloat[Varying number of robots]{%
    \input{plots/robot_ablation.tex}
    \label{fig:fair_coverage_ablation}
-  }\hfill
  \subfloat[Varying number of IDFs]{%
    \input{plots/constraint_ablation.tex}
    \label{fig:features_fair_coverage_ablation}
  }
  \hfill
  \subfloat[Varying environment length and robots]{%
    \input{plots/env_ablation.tex}
    \label{fig:env_size_ratio_lpac_centralied_all_robots}
  }
  \caption{Scalability study for the Fair Coverage problem: The plots show the coverage cost ratio of LPAC to Centralized CVT (C-CVT) on the $y$-axis.
    The simulations are executed for $500$ time steps of $0.5 s$ each and dual updates every $25$ steps.
    The same model is used for all experiments without further training.
    The LPAC consistently outperforms C-CVT across all experiments.
    While the performance reduces when the number of robots~$n$ is scaled down, the LPAC performs significantly better as $n$ scales up.
  The performance for varying numbers of IDFs and environment sizes remains relatively stable.}
  \label{fig:ablation_comparison}
\end{figure*}

We perform scalability studies for the fair coverage control problem for varying numbers of robots, number of IDFs, and environment size, as shown in \fgref{fig:ablation_comparison}.
The $y$-axis shows the ratio of the maximum coverage cost of LPAC to the centralized CVT controller averaged over $100$ environments of size $1024\times 1024$; the values are less than $1$ as LPAC outperforms the centralized CVT controller.
In \fgref{fig:fair_coverage_ablation}, while the LPAC does not perform as well when the number of robots~$n$ is scaled down from $32$, the performance improves significantly as $n$ scales up.
Additionally, \tbref{tab:centralized_vs_lpac_fair_robot_ablation} shows that the LPAC outperforms the centralized CVT controller for more than $90\%$ of the environments.
\fgref{fig:features_fair_coverage_ablation} and \tbref{tab:centralized_vs_lpac_fair_robot_ablation} show the stability of LPAC as the number of IDFs is scaled up.
 The LPAC outperforms centralized CVT by at least $35\%$ in all cases.
 
\fgref{fig:env_size_ratio_lpac_centralied_all_robots} shows the stability in performance as the size of the environment and number of robots are jointly scaled up.
Once again, LPAC outperforms the centralized CVT algorithm by at least $35\%$, thus showing that LPAC scales in the coverage control problem size. 

\textbf{Constrained Coverage:}
We consider the constrained coverage control problem \eqref{eqn:multi_objective_coverage} and compare the centralized CVT against the proposed primal-dual algorithm with LPAC policy.
To the best of our knowledge, there does not exist another algorithm to compare against for the constrained multi-objective coverage control problem.
We considered the feasibility problem with $M=4$ IDFs, i.e. $\phi_0=0$, and sampled the constraint thresholds~$\alpha_m$ from a normal distribution $\alpha_m\sim \mc N(\mu,0.1)$, for $\mu\in\{0.3,\dots,0.9\}$.

\begin{figure}[tbp]
  \input{plots/constrained_plot.tex}
  \caption{Constrained coverage control with varying constraint levels ($\mu$):
  The LPAC (solid) outperforms the centralized CVT (dashed) in terms of the percentage of infeasible IDFs (circles) and problems (squares).\label{fig:constraint}}
\end{figure}

\fgref{fig:constraint} shows the percentage of infeasible constraints (or equivalently IDFs) and infeasible problems (at least one of the four IDFs is infeasible).
Except for $\mu=0.3$ with $45\%$, the percentage of infeasible constraints in the LPAC approach is always less than half that of the centralized CVT algorithm.
Furthermore, averaged across all $\mu$, LPAC is infeasible for $26\%$ less problems.


The extensive simulation results establish that the proposed primal-dual algorithm with LPAC policy performs significantly better baseline algorithms for both fair coverage and constrained coverage problems.

%% file: plots/comms_sensors.tex
\begin{tikzpicture}[trim axis right]
  \footnotesize
  \pgfplotstableread[col sep=comma]{./plots/data/comms_sensor_sizes.csv}{\ratiodata}
  \begin{axis}[
    width=1.05\columnwidth,
    height=0.7\columnwidth,
    xlabel={Time},
    x unit={s},
    ylabel={Coverage cost ratio ($\frac{\text{LPAC}}{\text{Centralized CVT}}$)},
    xlabel style={font=\footnotesize},
    ylabel style={font=\footnotesize, inner sep=-0.3pt},
    legend cell align={left},
    legend columns=3,
    legend style={
      font=\scriptsize,
      inner xsep=1pt,
      nodes={inner sep=1.5pt,text depth=0.2em},
      at={(1.0,1.0)},
      anchor=north east,
      draw=none,
    },
    scaled x ticks=false,
    xmin=-5, xmax=505,
    ymin=0.50, ymax=1.01,
    grid=both,
    xticklabel style={font=\scriptsize},
    yticklabel style={font=\scriptsize},
    ]
    \foreach \cr/\ltype in {192/solid, 256/loosely dashed, 320/loosely dotted} {
      \foreach \sr/\ctype in {64/mRed, 96/mBlue, 128/mGreen} {
        \edef\temp{\noexpand\addplot[smooth,thick,mark=none,\ctype,\ltype] table[x=Step, y=\cr_\sr, col sep=comma] {\noexpand\ratiodata};}
          \temp
          \edef\temp{\noexpand\addlegendentry{$(\cr,\sr)\ $}}
            \temp
          }
        }
      \end{axis}
    \end{tikzpicture}

%% file: plots/case_study.tex
\begin{tikzpicture}
  \clip (-0.45,-2.6) rectangle (0.95\columnwidth,0.255\columnwidth);
  \pgfplotstableread[col sep=comma]{./plots/data/env_72_fair_coverage.csv}{\casedata}
  \begin{groupplot}[
    group style={
      group size=3 by 2,
      x descriptions at=edge bottom,
      y descriptions at=edge left,
      horizontal sep=0.1cm,
      vertical sep=0.1cm,
    },
    tiny,
    grid=both,
    no markers,
    scaled x ticks=false,
    width=0.49\columnwidth,
    height=0.4\columnwidth,
    xmin=-5, xmax=505,
    xtick={0,200,400},
    ]
    \def\typestr{_obj_}
    \pgfplotsforeachungrouped \controller in {clair, LPAC, centralized} {
      \global\let\globalcontroller\controller
      \nextgroupplot[ymin=0.2, ymax=4.5]
      \foreach \iter/\ltype/\ctype in {3/solid/mRed, 2/loosely dashed/blue, 1/dotted/mGreen, 4/loosely dashdotted/mSteelGray} {
        \edef\temp{\noexpand\addplot[smooth,thick,\ctype,\ltype] table[x=Step, y=\globalcontroller\typestr\iter, col sep=comma] {\noexpand\casedata};}
        \temp
      }
    }
    \def\typestr{_lambda_}
    \pgfplotsforeachungrouped \controller in {clair, LPAC, centralized} {
      \global\let\globalcontroller\controller
      \nextgroupplot[ymin=0.1, ymax=0.5]
      \foreach \iter/\ltype/\ctype in {3/solid/mRed, 2/loosely dashed/blue, 1/dotted/mGreen, 4/loosely dashdotted/mSteelGray} {
        \edef\temp{\noexpand\addplot[smooth,thick,\ctype,\ltype] table[x=Step, y=\globalcontroller_lambda_\iter, col sep=comma] {\noexpand\casedata};}
        \temp
      }
    }
  \end{groupplot}
  \node[anchor=south] (cvt) at ($(group c1r1.north)-(0,0.05)$) {\scriptsize Clairvoyant CVT};
  \node[anchor=south] (cvt) at ($(group c2r1.north)-(0,0.05)$) {\scriptsize LPAC};
  \node[anchor=south] (cvt) at ($(group c3r1.north)-(0,0.05)$) {\scriptsize Centralized CVT};
  \node[anchor=north] (cvt) at ($(group c2r2.south)-(0,0.25)$) {\scriptsize Time $[s]$};
\end{tikzpicture}

%% file: plots/robot_ablation.tex
\begin{tikzpicture}[trim axis right]
  \tiny
  \pgfplotstableread[col sep=comma]{./plots/data/robot_ablation_ratio_vals.csv}{\ratiodata}
  \begin{axis}[
    width=0.33\textwidth,
    height=0.28\textwidth,
    scaled x ticks=false,
    xmin=-5, xmax=255,
    ymin=0.55, ymax=1.05,
    grid=both,
    scaled x ticks=false,
    colormap name=redblue,
    colorbar sampled,
    colorbar style={
      width=0.1cm,
      xshift=-0.3cm,
      ytick={8,16,24,32,40,48,56,64},
    },
    cycle list={[samples of colormap={8 of redblue}]},
    point meta min=8,
    point meta max=64,
    y tick label style={
      /pgf/number format/.cd,
      fixed,
      fixed zerofill,
      precision=2,
      /tikz/.cd
    },
    xtick={0,50,100,150,200,250},
    xlabel={Time},
    x unit={s},
    x label style={font=\scriptsize},
    ]
    \foreach \nconst in {8,16,24,32,40,48,56,64} {
      \addplot+[smooth,mark=none,solid] table[x=Step, y=\nconst, col sep=comma] {\ratiodata};
    }

  \end{axis}
\end{tikzpicture}

%% file: plots/constraint_ablation.tex
\begin{tikzpicture}[trim axis right]
  \tiny
  \pgfplotstableread[col sep=comma]{./plots/data/features_ratio_vals.csv}{\ratiodata}
  \begin{axis}[
    width=0.34\textwidth,
    height=0.28\textwidth,
    scaled x ticks=false,
    xmin=-5, xmax=255,
    ymin=0.55, ymax=1.05,
    grid=both,
    scaled x ticks=false,
    colormap name=redblue,
    colorbar sampled,
    colorbar style={
      width=0.1cm,
      xshift=-0.3cm,
      ytick={4,5,6,7,8,9,10},
    },
    cycle list={[samples of colormap={7 of redblue}]},
    point meta min=4,
    point meta max=10,
    y tick label style={
        /pgf/number format/.cd,
            fixed,
            fixed zerofill,
            precision=2,
        /tikz/.cd
    },
    xtick={0,50,100,150,200,250},
    xlabel={Time},
    x unit={s},
    x label style={font=\scriptsize},
    ]

    \foreach \nconst in {4,5,6,7,8,9,10} {
        \addplot+[smooth,mark=none,solid] table[x=Step, y=\nconst, col sep=comma] {\ratiodata};
    }
    
  \end{axis}
\end{tikzpicture}

%% file: plots/env_ablation.tex
\begin{tikzpicture}
  \clip (-0.5,-0.5) rectangle (5.15, 3.6);   
  \tiny
  \pgfplotstableread[col sep=comma]{./plots/data/env_ablation.csv}{\ratiodata}
  \begin{axis}[
    width=0.33\textwidth,
    height=0.28\textwidth,
    scaled x ticks=false,
    xmin=-5, xmax=255,
    ymin=0.55, ymax=1.05,
    grid=both,
    scaled x ticks=false,
    colormap name=redblue,
    colorbar sampled,
    colorbar style={
      width=0.1cm,
      xshift=-0.3cm,
      ytick={1254,1448,1620,1774,1916,2048},
      yticklabels={$1254$ $(48)$,$1448$ $(64)$,$1620$ $(80)$,$1774$ $(96)$,$1916$ $(112)$,$2048$ $(128)$},
      yticklabel style={font=\tiny, text width=4em},
    },
    cycle list={[samples of colormap={6 of redblue}]},
    point meta min=1254,
    point meta max=2048,
    y tick label style={
      /pgf/number format/.cd,
      fixed,
      fixed zerofill,
      precision=2,
      /tikz/.cd
    },
    xtick={0,50,100,150,200,250},
    xlabel={Time},
    x unit={s},
    x label style={font=\scriptsize},
    ]

    \foreach \nconst in {1254,1448,1620,1774,1916,2048} {
      \addplot+[smooth,mark=none,solid] table[x=Step, y=\nconst, col sep=comma] {\ratiodata};
    }

  \end{axis}
\end{tikzpicture}

%% file: plots/constrained_plot.tex
\begin{tikzpicture}
  \small
  \footnotesize
  \pgfplotstableread[col sep=comma]{./plots/data/constrained_data.csv}{\infdata}
  \begin{axis}[
    width=1.05\columnwidth,
    height=0.7\columnwidth,
    xlabel={Constraint Levels ($\mu$)},
    ylabel={Infeasibility (\%)},
    xlabel style={font=\footnotesize},
    ylabel style={font=\footnotesize, inner sep=0.0pt},
		legend style={font=\scriptsize, at={(0.735,1.00)},anchor=north},
    legend cell align={left},
    scaled x ticks=false,
    xmin=0.28, xmax=0.92,
    ymin=-2, ymax=102,
    grid=both,
    xticklabel style={font=\scriptsize},
    yticklabel style={font=\scriptsize},
    x tick label style={
        /pgf/number format/.cd,
            fixed,
            fixed zerofill,
            precision=2,
        /tikz/.cd
    }
    ]

    \addplot[mark=*, mark options={solid,fill=mOrange}, draw=mOrange, thick, dashed]
    table[x=mus, y=idfs_cen, col sep=comma] {\infdata};
    \addlegendentry{Centralized CVT (constraints)}

    \addplot[mark=square*, mark options={solid,fill=mDarkRed}, draw=mDarkRed, thick, dashed]
    table[x=mus, y=prob_cen, col sep=comma] {\infdata};
    \addlegendentry{Centralized CVT (problems)}

    \addplot[mark=*, mark options={solid, fill=mBlue}, draw=mBlue, thick]  
    table[x=mus, y=idfs_lpac, col sep=comma] {\infdata};
    \addlegendentry{LPAC (constraints)}

    \addplot[mark=square*, mark options={solid,fill=mGreen}, draw=mGreen, thick]  
    table[x=mus, y=prob_lpac, col sep=comma] {\infdata};
    \addlegendentry{LPAC (problems)}

  \end{axis}
\end{tikzpicture}

%% file: conclusions.tex
In this paper, we reformulate the multi-objective coverage control with constraints as a coverage control problem with a single objective. Our method relies on utilizing a dynamically re-weighted combination of maps using a primal-dual algorithm. We showcased the benefits of our method by utilizing a decentralized learning technique based on Graph Neural Networks that abstracts the perceptions and learns the communications and actions to be taken by each robot based solely on local perception and communications.
Empirically, we implemented a wide variety of experiments showcasing that our method is scalable---both in the number of importance fields, as well as in the number of robots in the swarm---and improves upon both centralized and decentralized CVT algorithms by an average of $30\%$.
Future work involves real-world experiments in decentralized settings.